\documentclass[12pt]{article}
\pdfoutput=1

\newenvironment{proof}{{\em Proof:}}{\hfill{\hfill\rule{2mm}{2mm}}\\}

\usepackage[lined,ruled,linesnumbered,longend]{algorithm2e}
\usepackage{url,hyperref}
\usepackage{theorem,amsmath,amssymb,latexsym}
\usepackage[T1]{fontenc}
\usepackage{xspace}
\usepackage{color}

\newtheorem{theorem}{Theorem}
\newtheorem{corollary}[theorem]{Corollary}

\newcommand{\reals}{{\rm I\!\hspace{-0.025em} R}}
\newcommand{\mNNG}{\ensuremath{\textrm{NNG}}}
\newcommand{\mDT}{\ensuremath{\textrm{DT}}}

\begin{document}

\title{
	Delaunay Triangulations in Linear Time? (Part I) \thanks{This was originally part of a manuscript containing further material now in~\cite{b-cdts-09}. For the problem at hand, we found an algorithm avoiding the use of the history of the construnction~\cite{bm-dtotm-09}. This research was supported
by the Deutsche Forschungsgemeinschaft within the European graduate program 'Combinatorics, Geometry, and Computation' (No.\ GRK 588/2) and
by the Netherlands' Organisation for Scientific Research (NWO) under BRICKS/FOCUS project no. 642.065.503.  }
 }

\author{
      {Kevin Buchin} \thanks{Dept.\ of Information and Comp.\ Sci., Utrecht Univ., Netherlands;
\textsl{buchin@cs.uu.nl}.}
}

\date{}

\maketitle

\begin{abstract}
We present a new and simple randomized algorithm for constructing the Delaunay triangulation using nearest-neighbor graphs for point location. Under suitable assumptions, it runs in linear expected time for points in the plane with polynomially bounded spread, i.e., if the ratio between the largest and smallest pointwise distance is polynomially bounded. This also holds for point sets with bounded spread in higher dimensions as long as the expected complexity of the Delaunay triangulation of a sample of the points is linear in the sample size.
\end{abstract}

%\section{Introduction}
Chan and P{\v a}tra{\c s}cu~\cite{cp-pl-0x, cp-vd-07} presented
$o(N \log N)$ randomized algorithms for constructing Voronoi Diagrams of points in the plane (from which the Delaunay triangulation can be computed in linear time and vice-versa) under suitable models of computation. Here we present an $O(N)$ randomized algorithm for the Delaunay triangulation in the plane in a different model.
The algorithm is not restricted to two dimensions and it runs in linear expected time as long as the expected complexity of the Delaunay triangulation of a random sample of the input points is linear in the sample size. An example of linear complexity Delaunay triangulation are suitably sampled $(d-1)$-dimensional polyhedra in $\reals^d$. Our algorithm locates points by combining the history (i.e., the Delaunay tree~\cite{BT86,bt-rcdt-93}) of a randomized incremental construction with a sequence of nearest-neighbor graph computations. For the nearest-neighbor graphs we use a recent result by Chan~\cite{c-wspd-08} that links well-separated pair decompositions to sorting. By the use of radix sort this results in a linear time algorithm for well-separated pair decompositions and as a consequence for nearest-neighbor graphs. We will use the same assumptions as Chan on the model of computation and the point set. The model of computation is a real-RAM with a constant-time \emph{restricted} floor function that can be applied only if the resulting integer has $O(\log N)$ bits. Restricting the floor function avoids issues about creating an unreasonably powerful model of computation. The input point set should have polynomially bounded spread, i.e., the ratio of the largest and smallest point to point distance should be bounded by a polynomial in the size of the point set. But also other combinations of models of computation and sorting algorithms can be used, resulting in a optimal running time asymptotically bounded by the time needed for sorting (see~\cite{c-wspd-08} for details).

%\section{Incremental Constructions Using Nearest-Neighbor Graph}
%\subsection{Description}
\begin{figure*}
\linesnotnumbered
\begin{algorithm}[H]
\dontprintsemicolon
%\SetKwFor{ForEach}{In each}{do}{end}
\KwIn{Finite point set $P$ in $\reals^d$}
\KwOut{Delaunay triangulation of $P$}
\BlankLine
\nlset{1}Split $P$ into rounds $R_1,\ldots R_m$ of doubling size with $R_1$ of constant size, let $S_j := \bigcup_{1\leq i\leq j} R_i$ $(j=1\ldots m)$.\;
\nlset{2}Insert points in $R_1$ into the Delaunay triangulation using history for point location.\;
\nlset{3}For $k=2,\ldots,m$ insert points in $R_k$ into the Delaunay triangulation as follows:\;
\Indp
\nlset{3.1} Set $T_{k-1}\leftarrow R_k$, $T_i \leftarrow \emptyset$ ($0\leq i <k-1$), and $j \leftarrow k-1$.\;
\nlset{3.2} While $T_j \neq \emptyset$ and $j>0$: compute $\mNNG (T_j \cup S_j)$ and from each connected component with no vertex in $S_j$ add the first point of the component to $T_{j-1}$; then set $j\leftarrow j-1$.\;
%\nlset{3.3} If $j=0$ and $T_0 \neq \emptyset$: Locate $T_0$ in $\mDT (S_1)$ using the history; j \leftarrow 1\;
\nlset{3.3} While $j<k-1$: locate $T_j$ (if not empty) in $\mDT (S_{j+1})$ using history starting at $\mDT (S_j)$; then locate $T_{j+1}$ in $\mDT (S_{j+1})$ by walking through the connected components starting at an already located point; then set $j\leftarrow j+1$\;
% For $i=j+1,\ldots,k-1$: locate $T_{i-1}$ in $\mDT (S_i)$ using history starting at $\mDT (S_{i-1})$; then locate $T_i$ in $\mDT (S_i)$ by walking through the connected components starting at an already located point.\;
%\Indm
\caption{Incremental Construction using Nearest-Neighbor Graph}
\label{alg:nng}
\end{algorithm}
\linesnumbered
\vspace*{-0.6cm}
\end{figure*}

\paragraph{General Setup.}
We construct the Delaunay triangulation of a finite point set $P\subset \reals^d$ by a randomized incremental construction using a history. The point location is accelerated by locating points at intermediate levels in the history instead of the top, see Algorithm~\ref{alg:nng}.
Given an insertion order we group the points into \emph{rounds} $R_1,\ldots, R_m$ in accordance with the order, i.e., the points in $R_i$ are in the insertion order before the points in $R_{i+1}$ for $1\leq i < m$. The rounds double in size, i.e., $|R_1|$ is constant, and $|R_{i+1}| = 2 |R_i|$ (with possibly the exception of the last round for which $|R_m| \leq 2 R_{m-1}$). Let $S_j := \bigcup_{1\leq i\leq j} R_i$ denote the points inserted in or before round $j$.
Together with the history graph we store the Delaunay triangulations of the $S_j$. Note that the rounds are only used for facilitating the point location; the insertion order remains the same.
%Points are inserted in rounds of doubling size, the first round having constant size.

\paragraph{Point Location in a Round.}
The points of the first round are located in the standard way using the history. At the beginning of round $k$ ($2\leq k \leq m$) the points of the round $R_k$ are located in the Delaunay triangulation of $S_{k-1}$ using a family of sets $T_1, \ldots, T_{k-1}$ (in every round a different family, thus more formally the family could be written as $T_1^k, \ldots, T_{k-1}^k$):
%At the beginning of a round, we locate the points in the following way.
Let $T_{k-1} := R_k$.
We compute the nearest-neighbor graph of $T_{k-1} \cup S_{k-1}$. For connected components of the nearest-neighbor graph without a vertex in $S_{k-1}$ we include the first point (according to the insertion order) of the component in a set $T_{k-2}$.
We repeat the same procedure higher up in the history, i.e., we compute the nearest-neighbor graph of
$T_{k-2}\cup S_{k-2}$, for each connected component without a vertex in $S_{k-2}$ we include the first point in a set $T_{k-3}$, and so on. We stop this process with the construction of $T_0$ (or earlier with $T_j$ if $T_{j-1}$ is empty. For simplicity we describe the algorithm for the case that $T_0$ is not empty).
This yields a hierarchy of sets $T_{k-1}\supset T_{k-2}\supset \dots\supset T_0$.

Now we locate the points in $T_0$ in $\mDT (S_1)$ by using the history, i.e., we use the history to find a conflicting simplex and then locally search for the simplex containing $T_0$. For locating $T_1$ we have the following situation: each connected component of the nearest-neighbor graph of $T_1 \cup S_1$ either has a vertex in $S_1$ or has a vertex in $T_0$, thus each component has a vertex already located in $\mDT (S_1)$. We traverse each component starting at an already located vertex, e.g., by a depth first search. During the traversal we locate the traversed points in $\mDT (S_1)$ by \emph{walking} from an already located neighbor, i.e., we locally traverse the triangulation along the line segment between the two points.
After locating the points in $T_1$ in $\mDT (S_1)$ we locate them in $\mDT (S_2)$ by using the history, starting not at the top of the history but at the simplices of $\mDT (S_1)$ containing the points. By the same procedure we locate the points in $T_2$ in $\mDT (S_2)$ and $\mDT (S_3)$, and so on, until we have located the points in $T_{k-1} = R_k$ in $\mDT (S_{k-1})$. Finally we insert the points of $R_k$ into the Delaunay triangulation, where a new point is located using the history starting at the simplex of $\mDT (S_{k-1})$ containing the point.

% For connected components of the graph that have at least one point already inserted, we locate all new points by a traversing the component starting at an inserted point. For all other connected components, we include the first point (according to the insertion order) of the component in a set $T_{k-2}$.
% \emph{also add some remarks at the end, what and what not is essential}

%\subsection{Analysis}
\paragraph*{Analysis.}
In the analysis of the algorithm we will assume that the expected complexity of the Delaunay triangulation of a random sample of the given point set is linear in the size of the sample. This is the case for points in the plane, but it is also a realistic assumption for points sampled from a $(d-1)$-dimensional surface in $\reals^d$. The analysis can be adapted to the case where this assumption does not hold, yielding additional terms depending on the complexity of the triangulation.
In the following theorem we bound the run-time in terms of the cost of computing the nearest-neighbor graph. Note that
this bound holds for the standard real-RAM model and with no assumption about the spread of the point set.

\begin{theorem}
Let $P \subset \reals^d$ be a set of $N$ points in general position such that the expected complexity of the Delaunay triangulation of a random sample $R$ of $P$ of size $r$ is in $O(r)$. Algorithm~\ref{alg:nng} constructs the Delaunay triangulation of $P$ given in a random order in expected time
$ F(N) + \sum_{i=1}^{m-2}{(m-i) F( 2^i c)} + O(N)$,
where $c$ is the (constant) size of the first round, $m= \lceil \log_2 (N/c + 1) \rceil$, and
%
% $
% O( N+\sum_{i=0}^{\lfloor \log_2 (N) \rfloor}{(i+1) F( N/2^i)} ),
% $
$F(k)$ denotes the time needed to compute the nearest-neighbor graph of a subset of $P$ of size $k$.
% \[
% O(\sum_{i=0}^{\lfloor \log_2 (N) \rfloor}{(i+1) \left[ N\left( \frac{N}{2^i}\right) + \frac{N}{2^{i+1}}L\left( \right]
% \]
\end{theorem}

\begin{proof}
%We refer to Algorithm~\ref{alg:nng}  the
We will analyse the cost of Step~3.3. Step~1 is only a conceptual step and Step~2 takes constant time. Step~3.1 takes constant time per loop (of Step~3). For a given $k$ the nearest-neighbor graphs of $T_{k-1}\cup R_{k-1}$, $T_{k-2}\cup R_{k-2}$, \ldots, $T_1\cup R_1$ (or possibly fewer) are computed in Step~3.2. The size of these sets are bounded by
%$N/2^{m-k}, N/2^{m-k+1}$,\ldots,$N/2^{m-1}$, where $m = \log \lceil N/|R_1| + 1\rceil$.
$2^kc,2^{k-1}c,\ldots,2c$, where $c=|R_1|$ (except for $k-1 = m-1$, where $|T_{k-1}\cup R_{k-1}| = N$). Summing up over the loop of Step~3 this yields a cost of $F(N) + \sum_{i=1}^{m-2}{(m-i) F( 2^i c)}$ with $m= \lceil \log_2 (N/c) +1 \rceil$.

We now bound the cost of Step~3.3 for a given round $(k>1)$. The size of $R_k$ is $c2^{k-1}$. It suffices to prove that the cost of Step~3.3 is in $O(|R_k|)$. For this we construct sets $T_{k-1}',\ldots,T_0'\subset P$ such that any point of $R_k$ has the same probability to be included into $T_i \cup T_i'$ $(0\leq i \leq k-1)$. Let $T_{k-1}':= \emptyset$. We construct $T_i'$ $(k-1>i\geq 0)$ as follows: First we add each point of $T_{i+1}'$ with probability $1/2$. Second from each connected component in $\mNNG (T_{i+1}\cup S_{i+1})$ with a vertex in  $S_{i+1}$ we add each point of $T_{i+1}$ with probability $1/2$. Note that the choices do not need to be independent. For all other connected components we add each point of $T_{i+1}$ excluding the first and second (in the insertion order) with probability $1/2$. Note that for a connected component all points have the same probability of being the first point of the component (in this case it is included in $T_i \subset T_i\cup T_i'$). Likewise all points have the same probability of being the second point (in this case it is not included), and likewise the same probability that it is one of the remaining points (in this case it is included with probability $1/2$). Overall we get for all $i$ that any point of $R_k$ is included into $T_i\cup T_i'$ with probability $2^{i-k+1}$ and the expected size of $T_i\cup T_i'$ is $|R_k|2^{i-k+1}$.

We first bound the cost of locating a point $p \in T_j$ in $\mDT (S_{j+1})$. Since in the previous step we located $p$ in $\mDT (S_j)$, we can locate $p$ using the history starting at a conflicting simplex of $p$ in $\mDT (S_j)$. Since $S_{j+1}$ is a random subset of $P$ and the points of $S_{j+1}$ were inserted in a random order, the expected cost of locating $p$ would be $O(\log (|S_{j+1}|/|S_j|))=O(1)$ if $p$ were a random point of $R_k$~\cite{d-rysa-92}. This is not the case, but for a random point of $T_j \cup T_j'$ it would be the case. The cost of locating all points of $T_j$ in $\mDT (S_{j+1})$ is bounded by the cost of locating all points of $T_j \cup T_j'$ in $\mDT (S_{j+1})$ (knowing a conflict in $\mDT (S_j)$ for each point). The expected cost of this is in $O(E(T_j \cup T_j')) = O(2^{j-k+1}|R_k|)$.

This gives us the expected cost of locating one conflicting simplex for each point $p \in T_j$. We actually need to find the simplex in $\mDT (S_{j+1})$ containing the point $p$. This can be done by locally searching all conflicting simplices, one of which contains the point. The cost of this is therefore proportional to the number of conflicts a point has with $\mDT (S_{j+1})$. The cost for all points in $T_j$ can again be bounded by the total number of conflicts of $T_j \cup T_j'$ with $\mDT (S_{j+1})$ which is expected to be in $2^{j-k+1}O(|R_k|)$.

Second we bound the cost for locating the points of $T_{j+1}$ in $\mDT (S_{j+1})$. For locating a point of $T_{j+1}$, we traverse the connected components of $\mNNG (T_{j+1} \cup S_{j+1})$. For each component we start at a point for which we know the location in $\mDT (S_{j+1})$, i.e., a point from $S_{j+1} \cup T_j$. Assume we traverse the edge between $p$ and $q$ where $p$ is already located and $q$ needs to be located.
The point $q$ is located by walking along the line segment $pq$, i.e., by traversing the Delaunay triangulation along $pq$. The cost corresponds to the number of intersected simplices. Any simplex intersected is either in conflict with $p$ or with $q$, i.e., $p$ or $q$ lie in its circumsphere. If $p \in S_{j+1}$ we additionally have the cost of searching for the $d$-simplices adjacent to $p$ that is the first simplex on the walk.

Therefore the cost of the walk from $p$ to $q$ is bounded by the number of conflicts of $q$ with simplices of $\mDT (S_{j+1})$ and -- depending on whether $p$ is in $T_{j+1}$ or $S_{j+1}$ -- by the number of conflicts of $p$ or the number of faces at $p$. Any point can occur at most once as end point of a walk.
Furthermore since the degree of $\mNNG$ is in any fixed dimension bounded by the kissing number~\cite{cs-splg-98}, any point occurs only a constant times as starting point of a walk.
Thus the total cost of walking is up to a constant factor bounded by the complexity of $\mDT (S_{j+1})$ and the expected total number of conflicts of $T_{j+1}$ with $\mDT (S_{j+1})$. By assumption the expected complexity of $\mDT (S_{j+1})$ is linear in $|S_{j+1}|= (2^{j+1}-1)c$. The expected number of conflicts of $T_{j+1}$ with $\mDT (S_{j+1})$ we again bound by the expected number of conflicts of $T_{j+1}'$ with $\mDT (S_{j+1})$, which is $O(2^{j-k+2}|R_k|)$. Therefore the total expected cost of the round $R_k$ is in $O(\sum_{j=1}^{k-1} 2^{j-k+1}|R_k|) = O(|R_k|)$. Summing up over all rounds yields an expected linear cost.
\end{proof}

We now use that the nearest-neighbor graph of a point set with bounded spread can be computed in linear time. Note that the condition on the complexity of the Delaunay triangulation always holds in the plane, thus the algorithm computes the Delaunay triangulation of points in the plane with bounded spread in linear expected time.
\newpage
\begin{corollary}
Let $P \subset \reals^d$ be a set of $N$ points in general position with bounded spread such that the expected complexity of the Delaunay triangulation of a random sample $R \subset P$ of size $r$ is in $O(r)$. Algorithm~\ref{alg:nng} constructs the Delaunay triangulation of $P$ given in a random order in linear expected time on a real-RAM with a constant-time floor function restricted to $\log N$ bits.
\end{corollary}
We would like to note that the analysis can be extended to the case where the Delaunay hierarchy~\cite{d-dh-02} is used instead of a history (the hierarchy is then built level by level) and also to the case of biased randomized insertion orders~\cite{acr-iccb-03,k-ops-07}.

%%%%%%%%%%%%%%%%%%%%%%%%%%%%%%%%%%%%%%%%%%


\begin{thebibliography}{10}

\bibitem{acr-iccb-03}
N.~Amenta, S.~Choi, and G.~Rote.
\newblock Incremental constructions con {BRIO}.
\newblock In {\em Proc. 19th Annu. ACM Sympos. Comput. Geom.}, pages 211--219.
  ACM Press, 2003.

\bibitem{BT86}
J.-D. Boissonnat and M.~Teillaud.
\newblock A hierarchical representation of objects: the {D}elaunay tree.
\newblock In {\em Proc. 2nd}, pages 260--268, 1986.

\bibitem{bt-rcdt-93}
J.-D. Boissonnat and M.~Teillaud.
\newblock On the randomized construction of the {D}elaunay tree.
\newblock {\em Theor. Comput. Sci.}, 112(2):339--354, 1993.

\bibitem{k-ops-07}
K.~Buchin.
\newblock {\em Organizing Point Sets: Space-Filling Curves, {D}elaunay
  Tessellations of Random Point Sets, and Flow Complexes}.
\newblock PhD thesis, Free University Berlin, 2007.
\newblock
  http://www.diss.fu-berlin.de/diss/receive/FUDISS\_thesis\_000000003494.

\bibitem{b-cdts-09}
K.~Buchin.
\newblock Constructing {D}elaunay triangulations along space-filling curves.
\newblock In {\em Proc. 17th Annual European Symposium on Algorithms (ESA)}, pages 119--130, 2009.

\bibitem{bm-dtotm-09}
K.~Buchin and W.~Mulzer.
\newblock Delaunay triangulations in $O(\text{sort}(n))$ time and more.
\newblock In {\em Proc. 50th Annual Symposium on Foundations of Computer Science (FOCS)}, pages 139--148, 2009.


\bibitem{c-wspd-08}
T.~M. Chan.
\newblock Well-separated pair decomposition in linear time?
\newblock {\em Inf. Process. Lett.}, 107(5):138--141, 2008.

\bibitem{cp-pl-0x}
T.~M. Chan and M.~P{\v a}tra{\c s}cu.
\newblock Transdichotomous results in computational geometry, {I}: {P}oint
  location in sublogarithmic time.
\newblock {\em SIAM J. Comput.}
\newblock To appear. Preliminary versions in: Proc. 47th IEEE Sympos. Found.
  Comput. Sci. (FOCS), 2006, pp. 325--332, 333--342.

\bibitem{cp-vd-07}
T.~M. Chan and M.~P{\v a}tra{\c s}cu.
\newblock Voronoi diagrams in {$n\cdot 2^{O(\sqrt{\lg\lg n})}$} time.
\newblock In {\em Proc. 39th ACM Sympos. Theory of Computing (STOC)}, pages
  31--39, 2007.

\bibitem{cs-splg-98}
J.~H. Conway and N.~J.~A. Sloane.
\newblock {\em Sphere packings, lattices and groups}.
\newblock Springer-Verlag, 3rd edition, 1998.

\bibitem{d-rysa-92}
O.~Devillers.
\newblock Randomization yields simple $O(n \log^* n)$ algorithms for
  difficult $\Omega(n)$ problems.
\newblock {\em Int. J. Comput. Geometry Appl.}, 2(1):97--111, 1992.

\bibitem{d-dh-02}
O.~Devillers.
\newblock The {D}elaunay hierarchy.
\newblock {\em {Internat. J. Found. Comput. Sci.}}, 13:163--180, 2002.

\end{thebibliography}
\end{document}